\documentclass[3p,preprint,11pt]{elsarticle}

\journal{Communications in Nonlinear Science and Numerical Simulation}

\usepackage{amssymb,amsmath,amsthm,bm}

\newcommand\mymatrix[1]{\bm{\mathrm{#1}}}
\newtheorem{lemma}{Lemma}
\newtheorem{proposition}{Proposition}

\biboptions{sort&compress}

\usepackage{graphicx}
\graphicspath{{figures/}}

\newlength\imagewidth
\usepackage{color}
\definecolor{dgreen}{rgb}{0,.6,0}
\newcommand\modified[1]{\textcolor{dgreen}{#1}} % added or modified texts
\usepackage[normalem]{ulem}

\newcommand\Xkey{\textit{Xkey}}
\newcommand\Dkey{\textit{Dkey}}
\newcommand\CKS{\textit{CKS}}
\newcommand\key{\textit{key}}
\newcommand\HD{\textup{HD}}
\newcommand\VD{\textup{VD}}
\newcommand\mHD{\textup{mHD}}
\newcommand\mVD{\textup{mVD}}

\begin{document}

\begin{frontmatter}

\title{Breaking a modified substitution-diffusion image cipher based on chaotic standard and logistic maps}

\author[hk-polyu]{Chengqing Li\corref{corr}}
\ead{zjulcq@gmail.com}
\author[gm-constan]{Shujun Li\corref{corr}}
\ead[URL]{www.hooklee.com}
\author[hk-polyu]{Kwok-Tung Lo}
\cortext[corr]{Corresponding authors.}

\address[hk-polyu]{Department of Electronic and Information Engineering, The Hong Kong
Polytechnic University, Hong Kong, China}

\address[gm-constan]{Fachbereich Informatik und Informationswissenschaft,
Universit\"at Konstanz, Fach M697, Universit\"atsstra{\ss}e 10,
78457 Konstanz, Germany}

\begin{abstract}
Recently, an image encryption scheme based on chaotic standard and
logistic maps was proposed by Patidar et al. It was
later reported by Rhouma et al. that an
equivalent secret key can be reconstructed with only one
known/chosen-plaintext and the corresponding ciphertext.
Patidar et al. soon modified the original scheme
and claimed that the modified scheme is secure against Rhouma et
al.'s attack. In this paper, we point out that the modified scheme
is still insecure against the same known/chosen-plaintext attack.
In addition, some other security defects existing in both
the original and the modified schemes are also reported.
\end{abstract}

\begin{keyword}
cryptanalysis \sep known-plaintext attack \sep
chosen-plaintext attack \sep encryption \sep image \sep
chaos
\end{keyword}

\end{frontmatter}

\section{Introduction}

With the rapid development of information technology, multimedia
data are transmitted over all kinds of wired/wireless networks more
and more frequently. Consequently, security
of multimedia data becomes a serious concern in many
applications. However, traditional text
encryption schemes cannot be used in a naive way to
protect multimedia data efficiently in some applications, mainly due
to some special requirements of the whole
multimedia system. This challenge stirs the design of special
multimedia encryption schemes to become a hot research topic
in the past two decades. Because of the subtle similarity between
chaos and cryptography, a great number of multimedia encryption
schemes based on chaos have been presented
\cite{Chen&Yen:RCES:JSA2003, YaobinMao:CSF2004,
Flores:EncryptLatticeChaos06,Wong:ChaosEncrypt:IEEETCASII}.
Unfortunately, many of them have been found to have security
problems from the cryptographical point of view
\cite{Kaiwang:PLA2005,Li:AttackingMaoScheme2007,
Li:AttackingRCES2008, David:AttackingChaos08,Li:BreakImageCipher:IVC09}.
Some general rules about evaluating security of chaos-based encryption
schemes can be found in \cite{AlvarezLi:Rules:IJBC2006,
Li:ChaosImageVideoEncryption:Handbook2004}.

Since 2003, Pareek et al. have proposed a number of different
encryption schemes based on one or more chaotic maps
\cite{Pareek:PLA2003, Pareek:CNSNS2005, Pareek:ImageEncrypt:IVC2006,
Pareek:CNSNS2009}. Recent cryptanalytic results
\cite{Alvarez:PLA2003, Li:AttackingCNSNS2008, Li:AttackingIVC2009}
have shown that all the three schemes proposed in
\cite{Pareek:PLA2003, Pareek:CNSNS2005, Pareek:ImageEncrypt:IVC2006}
have security defects. In \cite{Pareek:CNSNS2009}, a new image
encryption scheme based on the logistic and standard maps
was proposed, where the two maps are used to generate a
pseudo-random number sequence (PRNS) controlling two kinds of
encryption operations. In \cite{Rhouma:BreakPareek:CNSNS09},
Rhouma et al. reported that the scheme is not
secure in the sense that an equivalent key can be obtained from only
one known/chosen plain-image and the corresponding cipher-image. To
resist Rhouma et al.'s attack, a modified
version of the original scheme was proposed in
\cite{Pareek:modify:CNSNS2009}. The present paper reports the
following findings: 1) the modified image encryption
scheme can still be broken by the same known/chosen-plaintext
attack under the same condition; 2) there are some other security
defects existing in both the modified and the original schemes.

The rest of this paper is organized as follows.
Section~\ref{sec:encryptscheme} briefly introduces the
image encryption schemes under study and the
known/chosen-plaintext attack reported in
\cite{Rhouma:BreakPareek:CNSNS09}. Our cryptanalytic
results are presented in Sec.~\ref{sec:cryptanalysis} in detail. The
last section concludes the paper.

\section{The image encryption schemes under study and Rhouma
et al.'s attack}
\label{sec:encryptscheme}

For both schemes, we make the following assumptions to
ease our description\footnote{To make the presentation more concise
and more consistent, some
notations in the original papers \cite{Pareek:CNSNS2009,
Pareek:modify:CNSNS2009} are also modified.}. The plaintext is a RGB true-color image of size $H\times W$
(height$\times$width), which can be denoted by an $H\times W$ matrix
of 3-tuple pixel values $\mymatrix{I}=\{I(i,j)\}_{0\leq
i\leq H-1 \atop 0\leq j\leq W-1}=\{(R(i,j), G(i,j),
B(i,j))\}_{0\leq i\leq H-1 \atop 0\leq j\leq W-1}$.
Similarly, the ciphertext corresponding to $\mymatrix{I}$
is denoted by $\mymatrix{I}'=\{I'(i,j)\}_{0\leq i\leq H-1
\atop 0\leq j\leq W-1}=\{(R'(i,j), G'(i,j), B'(i,j))\}_{0\leq i\leq
H-1 \atop 0\leq j\leq W-1}$. To further facilitate our
discussion, we adopt the terms in \cite{Pareek:modify:CNSNS2009}:
the original image encryption scheme is called PPS09 and the
modified one mPPS09.

\subsection{The original image encryption scheme PPS09 \cite{Pareek:CNSNS2009}}

\begin{itemize}
\item
\textit{Secret key}: three floating-point numbers $x_0$, $y_0$, $K$,
and one integer $N$, where $x_0$, $y_0\in(0, 2\pi)$, $K>18$,
$100<N<1100$.

\item
\textit{Initialization}: prepare data for encryption/decryption by
performing the following steps.
\begin{itemize}
\item
Generate four XORing keys as follows: $\Xkey(1)=\lfloor
256x_0/(2\pi)\rfloor$, $\Xkey(2)=\lfloor 256y_0/(2\pi)\rfloor$,
$\Xkey(3)=\lfloor K\bmod 256 \rfloor$, $\Xkey(4)=(N \bmod 256)$.
Then, generate a pseudo-image
$\mymatrix{I}_{\Xkey}=\{(R_{\Xkey}(i,j), G_{\Xkey}(i,j),
B_{\Xkey}(i,j))\}_{0\leq i\leq H-1 \atop 0\leq j\leq W-1}$ by
filling an $H\times W$ matrix with the four XORing keys repeatedly:
$R_{\Xkey}(i,j)=\Xkey((3k\bmod 4)+1)$,
$G_{\Xkey}(i,j)=\Xkey(((3k+1)\bmod 4)+1)$,
$B_{\Xkey}(i,j)=\Xkey(((3k+2)\bmod 4)+1)$, where $k=iW+j$.

\item
Iterate the standard map Eq.~(\ref{eq:standardmap}) from the initial
conditions $(x_0,y_0)$ for $N$ times to obtain a new chaotic state
$(x_0', y_0')$. Then, further iterate it for $HW$ more
times to get $HW$ chaotic states
$\{(x_i,y_i)\}_{i=1}^{HW}$.

\begin{equation}
\label{eq:standardmap}
\begin{cases}
x  =  (x+K\sin(y))\bmod(2\pi),\\
y  =  (y+x+K\sin(y))\bmod(2\pi),
\end{cases}
\end{equation}

\item
Iterate the logistic map Eq.~(\ref{eq:logistic}) from the
initial condition $z_0=((x_0'+y_0')\bmod 1)$ for $N$ times to get a
new initial condition $z_0'$. Then, further iterate it for
$HW$ times to get $HW$ chaotic states
$\{z_i\}_{i=1}^{HW}$.

\begin{equation}
\label{eq:logistic} z=4z(1-z).
\end{equation}

\item
Generate a pseudo-image
$\mymatrix{I}_{\CKS}=\{(R_\textit{\CKS}(i,j),G_\textit{\CKS}(i,j),B_\textit{\CKS}(i,j))\}_{0\leq
i\leq H-1 \atop 0\leq j\leq W-1}$ by filling its R, G and B channels
with the three chaotic key streams (CKS) $\{x_k\}_{k=1}^{HW}$,
$\{y_k\}_{k=1}^{HW}$ and $\{z_k\}_{k=1}^{HW}$:
$R_\textit{\CKS}(i,j)=\left\lfloor
256x_k/(2\pi)\right\rfloor$,
$G_\textit{\CKS}(i,j)=\left\lfloor
256y_k/(2\pi)\right\rfloor$,
$B_\textit{\CKS}(i,j)=\left\lfloor 256z_k\right\rfloor$,
where $k=iW+j+1$.
\end{itemize}

\item
\textit{Encryption procedure}: a simple concatenation of the
following four encryption operations.

\begin{itemize}
\item
\textit{Confusion I}: Mask the plain-image $\mymatrix{I}$
by $\mymatrix{I}_{\Xkey}$ to obtain
$\mymatrix{I}^{\star}$, i.e.,
$\mymatrix{I}^{\star}=\mymatrix{I}\oplus\mymatrix{I}_{\Xkey}$.

\item
\textit{Horizontal Diffusion (HD)}: Scan
$\mymatrix{I}^{\star}=\{I^{\star}(i,j)\}_{0\leq i\leq H-1 \atop
0\leq j\leq W-1}$ rowwise from the upper-left pixel to the
bottom-right one, and mask each pixel value
(except for the first one) by its predecessor in the
scan. Denoting the output of this step by
$\mymatrix{I}^*=\{I^*(i,j)\}_{0\leq i\leq H-1 \atop 0\leq j\leq
W-1}$, the HD procedure is described as follows: 1)
$I^*(0,0)=I^{\star}(0,0)$; 2) for $k=1,\ldots,HW-1$,
\begin{equation} I^*(i,j)=I^{\star}(i,j) \oplus
I^*(i',j'),\label{eq:HD}
\end{equation}
where $i=\lfloor k/W\rfloor$, $j=(k\bmod W)$, $i'=\lfloor
(k-1)/W\rfloor$, $j'=((k-1)\bmod W)$.

\item
\textit{Vertical Diffusion (VD)}: Scan
$\mymatrix{I}^*$ columnwise from the bottom-right pixel to the
upper-left one, and mask each pixel value (except for
the first one) by its predecessor in the scan.
Denoting the output of this step by
$\mymatrix{I}^{**}=\{R^{**}(i,j),G^{**}(i,j),B^{**}(i,j)\}_{0\leq
i\leq H-1 \atop 0\leq j\leq W-1}$, the VD procedure can be described
as follows: 1) $I^{**}(H-1,W-1)=I^*(H-1,W-1)$; 2) for
$k=HW-2,\ldots,0$,
\begin{equation}
I^{**}(i,j)=I^*(i,j)\oplus \overline{I^{**}(i',j')},\label{eq:VD}
\end{equation}
where $i=(k\mod H)$, $j=\lfloor k/H \rfloor$, $i'=((k+1)\mod H)$,
$j'=\lfloor (k+1)/H \rfloor$, and
\[
\overline{I^{**}(i',j')}=(G^{**}(i',j')\oplus
B^{**}(i',j'),R^{**}(i',j')\oplus B^{**}(i',j'),R^{**}(i',j')\oplus
G^{**}(i',j')).
\]

\item
\textit{Confusion II}: Mask the pixel values in
$\mymatrix{I}^{**}$ with $\mymatrix{I}_{\CKS}$ to get the ciphertext
$\mymatrix{I}'$, i.e.,
$\mymatrix{I}'=\mymatrix{I}^{**}\oplus\mymatrix{I}_{\CKS}$.
\end{itemize}

\item
\textit{Decryption procedure}: the simple reversion of the above
encryption procedure.
\end{itemize}

\subsection{Rhouma et al.'s attack \cite{Rhouma:BreakPareek:CNSNS09}}

Denoting the horizontal and vertical diffusion processes by HD and
VD, respectively, the encryption procedure of PPS09 can be
represented as follows:
\begin{equation}
\mymatrix{I}'=\VD(\HD(\mymatrix{I}\oplus\mymatrix{I}_{\Xkey}))\oplus\mymatrix{I}_{\CKS}.\label{eq:PPS09}
\end{equation}

In \cite{Rhouma:BreakPareek:CNSNS09}, Rhouma et al. showed that the
HD and VD processes are commutative with XOR operations:
\begin{eqnarray*}
\HD(\bm{X}\oplus\bm{Y}) & = & \HD(\bm{X})\oplus\HD(\bm{Y}),\\
\VD(\bm{X}\oplus\bm{Y}) & = & \VD(\bm{X})\oplus\VD(\bm{Y}).
\end{eqnarray*}
Therefore, Eq.~\eqref{eq:PPS09} is equivalent to the following one:
\begin{equation}
\mymatrix{I}'=\VD(\HD(\mymatrix{I}))\oplus
\VD(\HD(\mymatrix{I}_{\Xkey}))\oplus
\mymatrix{I}_{\CKS}.\label{eq:PPS09eq}
\end{equation}
Assuming $\mymatrix{I}_{\key}=\VD(\HD(\mymatrix{I}_{\Xkey}))\oplus
\mymatrix{I}_{\CKS}$, we can observe the following two important
facts:
\begin{enumerate}
\item
neither HD nor VD depends on the key;

\item
$\mymatrix{I}_{\key}$ does not depend on the plaintext
$\mymatrix{I}$ or the ciphertext $\mymatrix{I}'$.
\end{enumerate}
The above facts immediately lead to a conclusion:
$\mymatrix{I}_{\key}$ can be used as an equivalent key to encrypt
any plaintext of the same size $H\times W$ and decrypt any
ciphertext of size $H\times W$. A known/chosen-plaintext attack can
be easily mounted to derive $\mymatrix{I}_{\key}$ from a
known/chosen plaintext $\mymatrix{I}$ and its corresponding
ciphertext $\mymatrix{I}'$:
\begin{equation}
\mymatrix{I}_{\key}=\VD(\HD(\mymatrix{I}))\oplus\mymatrix{I}'.
\end{equation}

\subsection{The modified image encryption scheme mPPS09 \cite{Pareek:modify:CNSNS2009}}

To enhance the security of PPS09 against Rhouma et al.'s attack, in
\cite{Pareek:modify:CNSNS2009} Patidar et al. proposed a modified
edition of PPS09 by making both HD and VD dependent on the secret
key.

The modified key-dependent HD and VD processes are denoted by mHD
and mVD in \cite{Pareek:modify:CNSNS2009}. Both mHD and mVD are
based on 16 diffusion keys derived from the secret key
$(x_0,y_0,K,N)$:
\begin{itemize}
\item
for $i=1,\ldots,5$, $\Dkey(i)=\sum_{j=0}^{2}a_{3\cdot (i-1)+j}\cdot
10^{2-j}\bmod 256$, where $x_0=a_1.a_2\ldots a_{15}\ldots$ and $a_i$
are decimal digits representing $x_0$;

\item
for $i=6,\ldots,10$, $\Dkey(i)=\sum_{j=0}^{2}b_{3\cdot (i-6)+j}\cdot
10^{2-j}\bmod 256$, where $y_0=b_1.b_2\ldots b_{15}\ldots$ and $b_i$
are decimal digits representing $y_0$;

\item
for $i=11,\ldots,15$, $\Dkey(i)=\sum_{j=0}^{2}c_{3\cdot
(i-11)+j}\cdot 10^{2-j}\bmod 256$, where $K=\ldots c_1.c_2\ldots
c_{15}\ldots$ and $c_i$ are decimal digits representing $K$;

\item
$\Dkey(16)=(N \bmod 256)$.
\end{itemize}

The mHD process is modified from HD by replacing Eq.~\eqref{eq:HD}
with the following equation:
\begin{equation}
I^*(i,j)=I^{\star}(i,j) \oplus I^*(i',j')\oplus\Dkey^*(k-1),
\end{equation}
where
\[
\Dkey^*(k)=(\Dkey((k\bmod 16)+1),\Dkey((k\bmod 16)+1),\Dkey((k\bmod
16)+1)).
\]

The mVD process is modified from VD by replacing Eq.~\eqref{eq:VD}
with the following equation:
\begin{equation}
I^{**}(i,j)=I^*(i,j)\oplus \overline{I^{**}(i',j')}\oplus
\Dkey^{**}(k'),
\end{equation}
where $k'=HW-2-k$ and
\[
\Dkey^{**}(k')=(\Dkey(3k'\bmod 16)+1),\Dkey(((3k'+1)\bmod
16)+1),\Dkey(((3k'+2)\bmod 16)+1)).
\]

\section{Cryptanalysis}
\label{sec:cryptanalysis}

In this section, we first show that the key-dependent horizontal and
vertical diffusion steps mHD and mVD do not increase the security of
mPPS09 against Rhouma et al.'s attack. Then we point out some common
security weaknesses in both PPS09 and mPPS09.

\subsection{Insecurity of mPPS09 against Rhouma et al.'s
attack}

Although both mHD and mVD are dependent on the secret key, we
noticed that they can be represented in an equivalent form which
renders the key-dependence useless. Assuming $\mymatrix{X}$ is the
input matrix and $\mymatrix{\Theta}$ is a zero matrix of the same
size as $\mymatrix{X}$, we have the following two lemmas.

\begin{lemma}
$\mHD(\mymatrix{X})=\HD(\mymatrix{X})\oplus
\mHD(\mymatrix{\Theta})$.\label{lemma:HD-mHD}
\end{lemma}
\begin{proof}
This lemma can be easily proved with mathematical induction on $k$.

For $k=0$, i.e., $i=j=0$, we have $\mHD(X(0,0))=X(0,0)$ and
$\HD(X(0,0))\oplus \mHD(\Theta(0,0))=X(0,0)\oplus(0,0,0)=X(0,0)$.
This lemma holds. Then, assume the lemma is true for $k\geq 0$, let
us prove the case of $k+1$.

For $k+1$, i.e., $i=\lfloor(k+1)/W\rfloor$, $j=((k+1)\bmod W)$,
$i'=\lfloor k/W\rfloor$ and $j'=(k\bmod W)$,
$\mHD(X(i,j))=X(i,j)\oplus\mHD(X(i',j'))\oplus\Dkey^*(k)$. According
to the assumption on $k$, we have
$\mHD(X(i',j'))=\HD(X(i',j'))\oplus\mHD(\Theta(i',j'))$. Thus,
$\mHD(X(i,j))=X(i,j)\oplus\HD(X(i',j'))\oplus\mHD(\Theta(i',j'))\oplus\Dkey^*(k)$.
Noting that $\HD(X(i,j))=X(i,j)\oplus\HD(X(i',j'))$, we get
$\mHD(X(i,j))=\HD(X(i,j))\oplus\mHD(\Theta(i',j'))\oplus\Dkey^*(k)$.
Further note that
$\mHD(\Theta(i,j))=\Theta(i,j)\oplus\mHD(\Theta(i',j'))\oplus\Dkey^*(k)=\mHD(\Theta(i',j'))\oplus\Dkey^*(k)$.
This immediately leads to
$\mHD(X(i,j))=\HD(X(i,j))\oplus\mHD(\Theta(i,j))$.
\end{proof}

\begin{lemma}
$\mVD(\mymatrix{X})=\VD(\mymatrix{X})\oplus
\mVD(\mymatrix{\Theta})$.
\end{lemma}
\begin{proof}
This lemma can be proved in a similar way to
Lemma~\ref{lemma:HD-mHD}, but the mathematical induction should be
made in descending order on $k$ (starting from $k=HW-1$ and ending
at $k=0$).
\end{proof}

The above two lemmas lead to the following proposition.
\begin{proposition}
The encryption procedure of mPPS09 is equivalent to the following
equation:
\begin{equation}
\mymatrix{I}'=\VD(\HD(\mymatrix{I}))\oplus \tilde{\bm{I}}_{\key},
\label{eq:EncryptFunction}
\end{equation}
where $\tilde{\bm{I}}_{\key}=\VD(\HD(\mymatrix{I}_{\Xkey}))\oplus
\VD(\mHD(\Theta))\oplus \mVD(\Theta)\oplus \mymatrix{I}_{\CKS}$.
\end{proposition}
\begin{proof}
From the properties of HD \& VD and Lemmas 1 \& 2, we can make the
following deduction:
\begin{eqnarray*}
\mymatrix{I}' & = &
\mVD(\mHD(\mymatrix{I}\oplus\mymatrix{I}_{\Xkey}))\oplus\mymatrix{I}_{\CKS},\\
& = &
\mVD(\HD(\mymatrix{I}\oplus\mymatrix{I}_{\Xkey})\oplus\mHD(\Theta))\oplus\mymatrix{I}_{\CKS},\\
& = &
\VD(\HD(\mymatrix{I}\oplus\mymatrix{I}_{\Xkey})\oplus\mHD(\Theta))\oplus\mVD(\Theta)\oplus\mymatrix{I}_{\CKS},\\
& = &
\VD(\HD(\mymatrix{I}\oplus\mymatrix{I}_{\Xkey}))\oplus\VD(\mHD(\Theta))\oplus\mVD(\Theta)\oplus\mymatrix{I}_{\CKS},\\
& = &
\VD(\HD(\mymatrix{I}))\oplus\VD(\HD(\mymatrix{I}_{\Xkey}))\oplus\VD(\mHD(\Theta))\oplus\mVD(\Theta)\oplus\mymatrix{I}_{\CKS},\\
& = & \VD(\HD(\mymatrix{I}))\oplus\tilde{\mymatrix{I}}_{\key}.
\end{eqnarray*}
This proves the proposition.
\end{proof}

Since $\mHD(\Theta)$ and $\mVD(\Theta)$ are both independent of the
plaintext and the ciphertext, they are uniquely determined by the
key $(x_0,y_0,K,N)$. This means that $\tilde{\mymatrix{I}}_{\key}$
is also uniquely determined by the key $(x_0,y_0,K,N)$. Therefore,
$\tilde{\mymatrix{I}}_{\key}$ can be used as an equivalent key of
mPPS09 exactly in the same way as $\mymatrix{I}_{\key}$ in PPS09. In
fact, even the determination process of the equivalent key is also
the same:
\[
\tilde{\mymatrix{I}}_{\key}=\VD(\HD(\mymatrix{I}))\oplus\mymatrix{I}'.
\]
This means that the same known/chosen-plaintext attack can be
applied to mPPS09 without any change to the program. In other words,
the security of mPPS09 against Rhouma et al.'s attack remains the
same as that of the original scheme PPS09.

We have performed some experiments to verify the correctness of the
conclusion. With the secret key $(x_0, y_0, K,
N)=(3.98235562892545,$ $ 1.34536356538912, $ $108.54365761256745,$ $
110)$, the equivalent key $\tilde{\mymatrix{I}}_{\key}$
was constructed from a known plain-image ``Lenna'' and the
corresponding cipher-image, which are shown in
Figs.~\ref{figure:knownplaintextattack}a) and b), respectively.
Then, \modified{$\tilde{\mymatrix{I}}_{\key}$} was used to recover
\modified{a} cipher-image shown in
Fig.~\ref{figure:knownplaintextattack}c, and the
plain-image ``Peppers'' (Fig.~\ref{figure:knownplaintextattack}d)
was successfully recovered.

\begin{figure}[!htb]
\centering
\begin{minipage}[t]{\imagewidth}
\centering
\includegraphics[width=\textwidth]{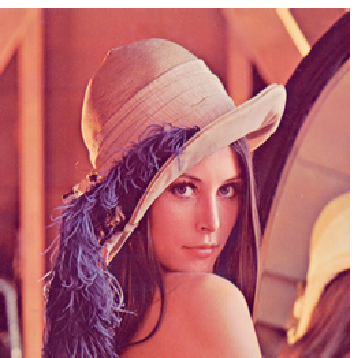}
a)
\end{minipage}
\begin{minipage}[t]{\imagewidth}
\centering
\includegraphics[width=\textwidth]{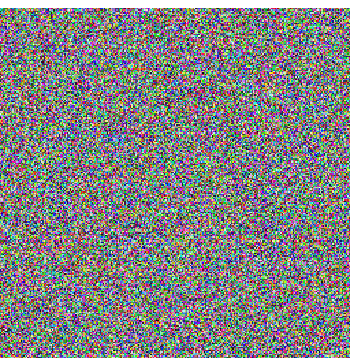}
b)
\end{minipage}\\
\begin{minipage}[t]{\imagewidth}
\centering
\includegraphics[width=\textwidth]{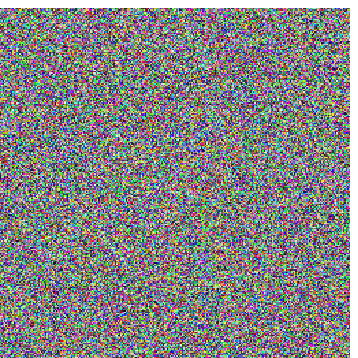}
c)
\end{minipage}
\begin{minipage}[t]{\imagewidth}
\centering
\includegraphics[width=\textwidth]{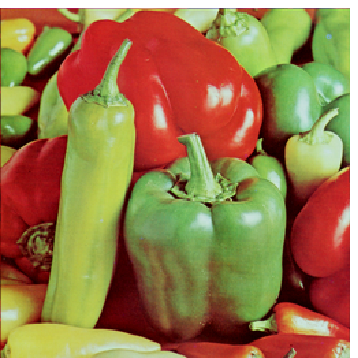}
d)
\end{minipage}
\caption{An experimental result of the proposed known-plaintext
attack: a) the known plain-image ``Lenna''; b) the corresponding cipher-image;
c) a cipher-image encrypted with the
same key; d) the recovered plain-image ``Peppers''.}
\label{figure:knownplaintextattack}
\end{figure}

\subsection{Other security weaknesses of PPS09 and mPPS09}

\subsubsection{Insufficient randomness of the PRNS $\{B_{\CKS}(i,j)\}$}

As illustrated in \cite{Li:AttackingBitshiftXOR2007}, the randomness
of pseudo-random bit sequences derived from chaotic orbits
of the logistic map is very weak. To further verify the randomness
of the PRNS $\{B_{\CKS}(i,j)\}$ generated via the
logistic map with control parameter 4.0, we
tested 100 PRNSs of length $512\times 512=262144$ (the number of
bytes used for encryption of a $512\times 512$ plain color image) by
using the NIST statistical test suite \cite{Rukhin:TestPRNG:NIST}.
The 100 sequences were generated with randomly selected
secret keys, and transformed to 1-D bit sequences by
concatenating the bits of all the elements. For each test, the
default significance level 0.01 was used. The results are shown in
Table~\ref{table:test}, from which one can see that the PRNS
$\{B_{\CKS}(i,j)\}$ is not random enough.

\begin{table}[!htbp]
\centering\caption{The performed tests with respect to a
significance level 0.01 and the number of sequences passing each
test in 100 randomly generated sequences.}\label{table:test}
\begin{tabular}{c|c}
\hline Name of Test                                    & Number of Passed Sequences\\
\hline\hline Frequency                                 & 95 \\
\hline Block Frequency ($m=100$)                       & 0  \\
\hline Cumulative Sums-Forward                         & 93 \\
\hline Runs                                            & 0  \\
\hline Rank                                            & 0  \\
\hline Non-overlapping Template ($m=9$, $B=010000111$) & 10 \\
\hline Serial ($m=16$)                                 & 0  \\
\hline Approximate Entropy ($m=10$)                    & 0  \\
\hline FFT                                             & 0  \\
\hline
\end{tabular}
\end{table}

\subsubsection{Insufficient sensitivity with respect to change of plaintext}

In \cite{Pareek:CNSNS2009,Pareek:modify:CNSNS2009},
Patidar et al. recognized that the sensitivity of
cipher-image with respect to change of plain-image is very
important. However, both PPS09 and mPPS09 are actually
very far from the desired property. As well known in cryptography,
this property is termed as avalanche effect. Ideally, it requires
the change of any single bit of plain-image will make every bit of
cipher-image change with a probability of one half.

For both PPS09 and mPPS09, the following equation holds for two
plain-images $\mymatrix{I}$ and
$\mymatrix{J}=\mymatrix{I}\oplus\mymatrix{I}_{\Delta}$:
\begin{eqnarray*}
\mymatrix{I}'\oplus\mymatrix{J}' & = &
(\VD(\HD(\mymatrix{I})))\oplus
(\VD(\HD(\mymatrix{J}))),\\
& = & \VD(\HD(\mymatrix{I}\oplus\mymatrix{J})),\\
& = & \VD(\HD(\mymatrix{I}_{\Delta})).
\end{eqnarray*}

The above equation implies the following two facts:
\begin{itemize}
\item
any change in a single bitplane will not change any other bitplanes
in the cipher-image;

\item
a change in plain-image $\mymatrix{I}_{\Delta}$ will cause a change
pattern determined by $\VD(\HD(\mymatrix{I}_{\Delta}))$, which is
far from a random pattern.
\end{itemize}

\color{black}

To show this defect clearly, we made an experiment by changing only
one bit of the red channel of a plain-image. It is found that only
some bits on the same bitplane in the corresponding cipher-image
were changed. The locations of the changed bits can be
seen from the differential cipher-image
$\VD(\HD(\mymatrix{I}_{\Delta}))$ and its three color channels as
shown in Fig.~\ref{figure:changebits}. Apparently, the change
pattern is far from random and balanced.

\begin{figure}[!htb]
\centering
\begin{minipage}[t]{\imagewidth}
\centering
\includegraphics[width=\textwidth]{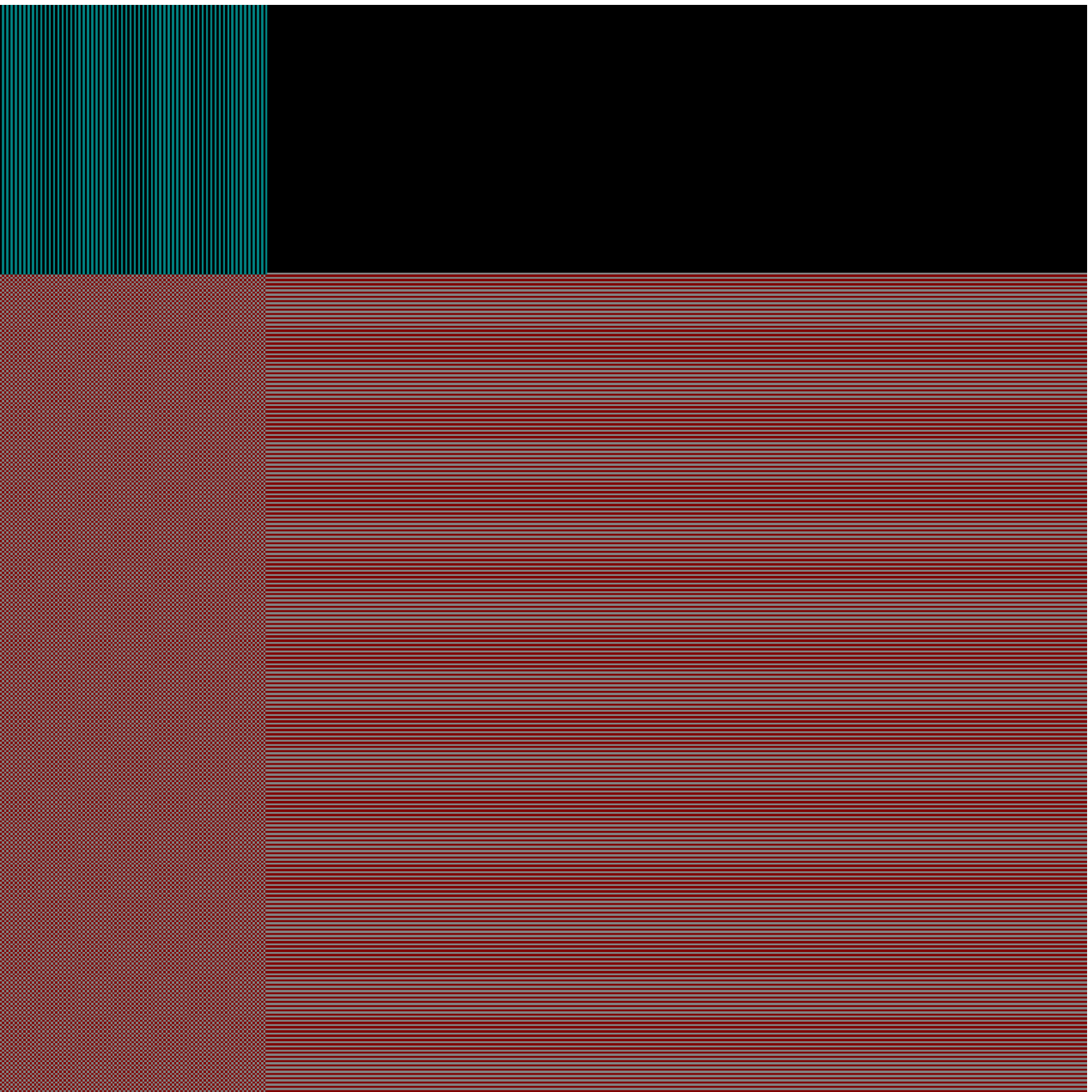}
a)
\end{minipage}\\
\begin{minipage}[t]{\imagewidth}
\centering
\includegraphics[width=\textwidth]{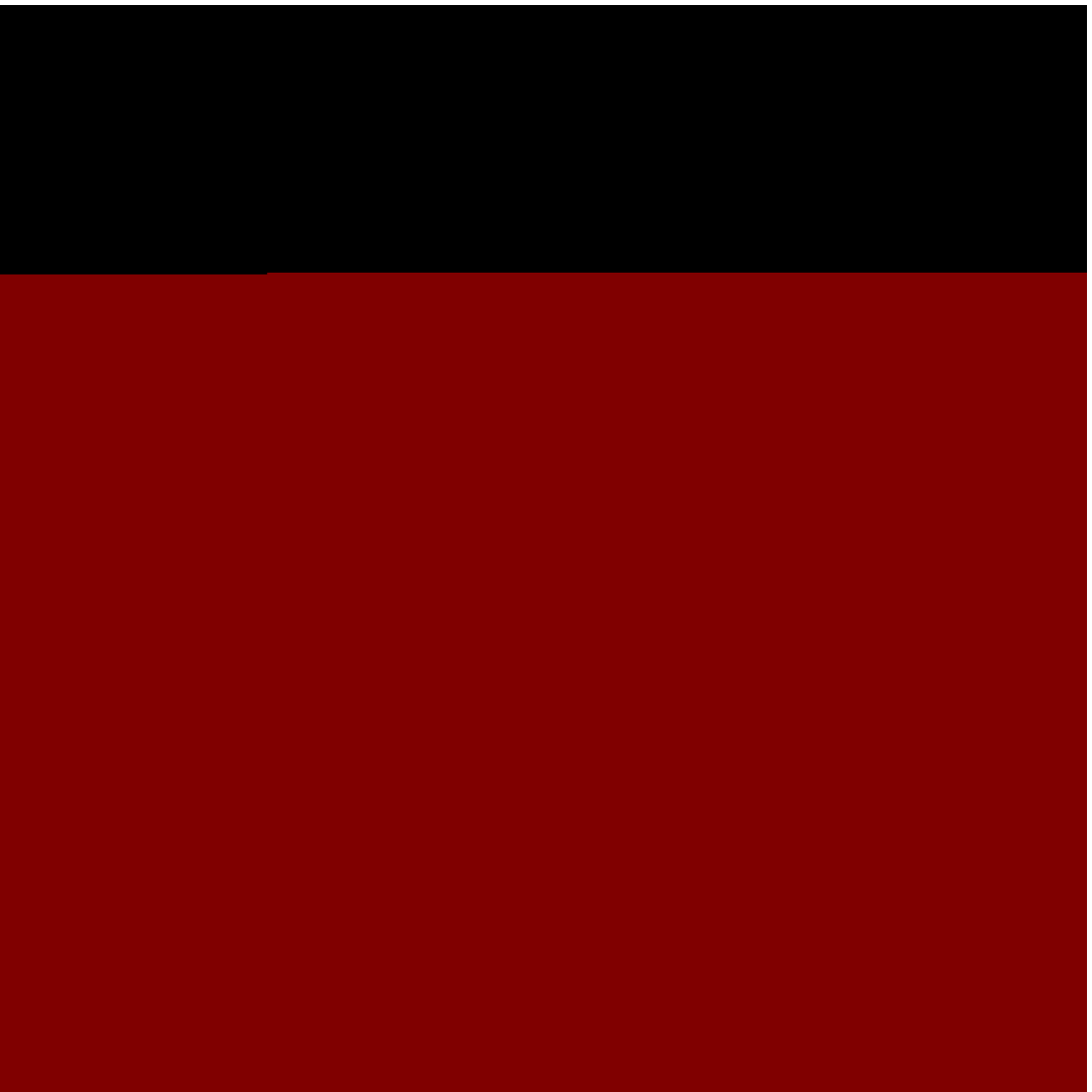}
b)
\end{minipage}
\begin{minipage}[t]{\imagewidth}
\centering
\includegraphics[width=\textwidth]{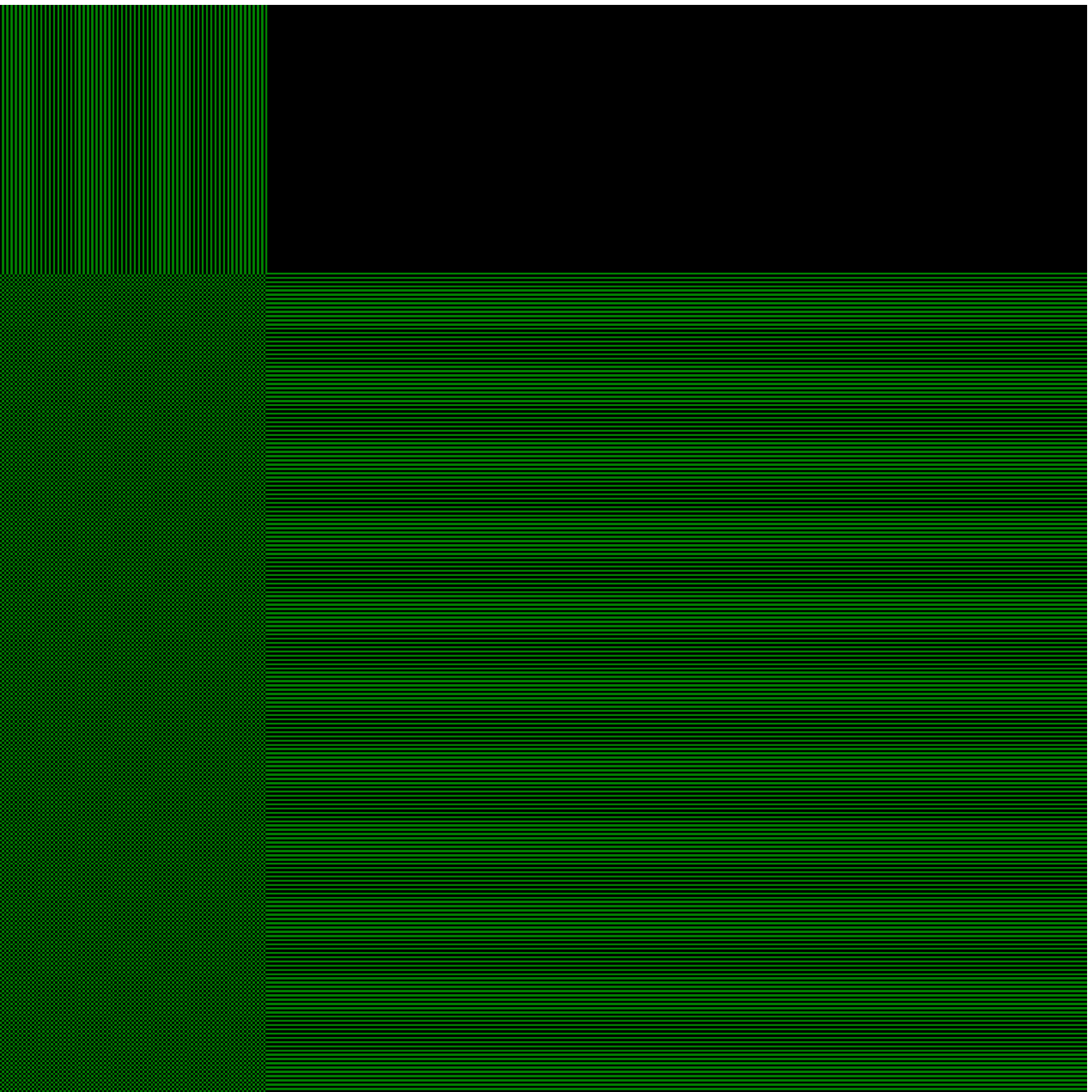}
c)
\end{minipage}
\begin{minipage}[t]{\imagewidth}
\centering
\includegraphics[width=\textwidth]{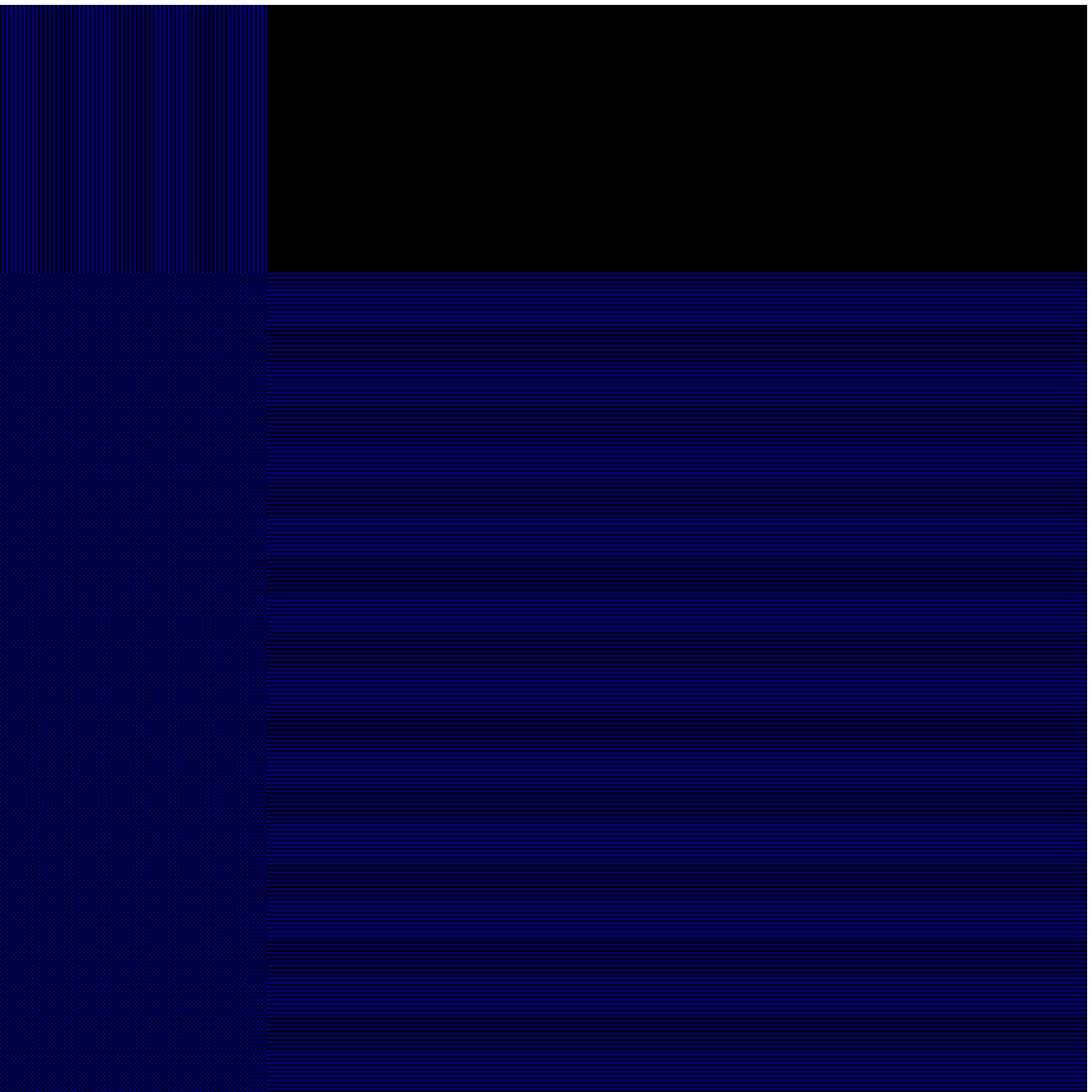}
d)
\end{minipage}
\caption{The differential cipher-image and its three color
channels, when the MSB (i.e., the 8-th bit) of $R(127, 127)$ in a
plain-image was changed: a) the differential cipher-image; b) red
channel; c) green channel; d) blue
channel.}\label{figure:changebits}
\end{figure}

\section{Conclusion}

In this paper, the security of the image encryption scheme
proposed in \cite{Pareek:modify:CNSNS2009} (a modified version of
the one proposed in \cite{Pareek:CNSNS2009}) is re-evaluated. It is
found that the scheme is still insecure against a
known/chosen-plaintext attack which can break the original scheme in
\cite{Rhouma:BreakPareek:CNSNS09}. In addition, two more security
weaknesses of both the original and the modified image encryption
schemes are reported: insufficient randomness of a PRNS involved,
and insufficient sensitivity with respect to change of plain-image.
Due to such a low level of security, we recommend not to use the
image encryption schemes under study unless their security
is further enhanced with more complicated countermeasures.

\section*{Acknowledgement}

Chengqing Li was supported by The Hong Kong Polytechnic University's
Postdoctoral Fellowships Scheme under grant no. G-YX2L. Shujun Li
was supported by a fellowship from the Zukunftskolleg of the
Universit\"at Konstanz, Germany, which is part of the
``Exzellenzinitiative'' Program of the DFG (German Research
Foundation).

\bibliographystyle{elsarticle-num}
%\nocite*
\bibliography{Pareek}

\begin{thebibliography}{10}
\expandafter\ifx\csname url\endcsname\relax
  \def\url#1{\texttt{#1}}\fi
\expandafter\ifx\csname urlprefix\endcsname\relax\def\urlprefix{URL }\fi
\expandafter\ifx\csname href\endcsname\relax
  \def\href#1#2{#2} \def\path#1{#1}\fi

\bibitem{Chen&Yen:RCES:JSA2003}
H.-C. Chen, J.-C. Yen, A new cryptography system and its {VLSI} realization,
  Journal of Systems Architecture 49~(7-9) (2003) 355--367.

\bibitem{YaobinMao:CSF2004}
G.~Chen, Y.~Mao, C.~K. Chui, A symmetric image encryption scheme based on {3D}
  chaotic cat maps, Chaos, Solitons \& Fractals 21~(3) (2004) 749--761.

\bibitem{Flores:EncryptLatticeChaos06}
N.~J. Flores-Carmona, M.~Carpio-Valadez, Encryption and decryption of images
  with chaotic map lattices, Chaos 16~(3) (2006) art. no. 033118.

\bibitem{Wong:ChaosEncrypt:IEEETCASII}
K.-W. Wong, C.-H. Yuen, Embedding compression in chaos-based cryptography, IEEE
  Transactions on Circuits and Systems II: Express Brief 55~(11) (2008)
  1193--1197.

\bibitem{Kaiwang:PLA2005}
K.~Wang, W.~Pei, L.~Zou, A.~Song, Z.~He, On the security of 3{D} cat map based
  symmetric image encryption scheme, Physics Letters A 343~(6) (2005) 432--439.

\bibitem{Li:AttackingMaoScheme2007}
C.~Li, G.~Chen, On the security of a class of image encryption schemes, in:
  Proceedings of 2008 IEEE Int. Symposium on Circuits and Systems, 2008, pp.
  3290--3293.

\bibitem{Li:AttackingRCES2008}
S.~Li, C.~Li, G.~Chen, K.-T. Lo, Cryptanalysis of the {RCES/RSES} image
  encryption scheme, Journal of Systems and Software 81~(7) (2008) 1130--1143.

\bibitem{David:AttackingChaos08}
D.~Arroyo, R.~Rhouma, G.~Alvarez, S.~Li, V.~Fernandez, On the security of a new
  image encryption scheme based on chaotic map lattices, Chaos 18~(3) (2008)
  art. no. 033112.

\bibitem{Li:BreakImageCipher:IVC09}
C.~Li, S.~Li, G.~Chen, W.~A. Halang, Cryptanalysis of an image encryption
  scheme based on a compound chaotic sequence, Image and Vision Computing
  27~(8) (2009) 1035--1039.

\bibitem{AlvarezLi:Rules:IJBC2006}
G.~\'{A}lvarez, S.~Li, Some basic cryptographic requirements for chaos-based
  cryptosystems, International Journal of Bifurcation and Chaos 16~(8) (2006)
  2129--2151.

\bibitem{Li:ChaosImageVideoEncryption:Handbook2004}
S.~Li, G.~Chen, X.~Zheng, Chaos-based encryption for digital images and videos,
  in: B.~Furht, D.~Kirovski (Eds.), Multimedia Security Handbook, CRC Press,
  2004, Ch.~4, pp. 133--167.

\bibitem{Pareek:PLA2003}
N.~Pareek, V.~Patidar, K.~Sud, Discrete chaotic cryptography using external
  key, Physics Letters A 309~(1-2) (2003) 75--82.

\bibitem{Pareek:CNSNS2005}
N.~Pareek, V.~Patidar, K.~Sud, Cryptography using multiple one-dimensional
  chaotic maps, Communications in Nonlinear Science and Numerical Simulation
  10~(7) (2005) 715--723.

\bibitem{Pareek:ImageEncrypt:IVC2006}
N.~Pareek, V.~Patidar, K.~Sud, Image encryption using chaotic logistic map,
  Image and Vision Computing 24~(9) (2006) 926--934.

\bibitem{Pareek:CNSNS2009}
V.~Patidar, N.~Pareek, K.~Sud, A new substitution-diffusion based image cipher
  using chaotic standard and logistic maps, Communications in Nonlinear Science
  and Numerical Simulation 14~(7) (2009) 3056--3075.

\bibitem{Alvarez:PLA2003}
G.~\'{A}lvarez, F.~Montoya, M.~Romera, G.~Pastor, Cryptanalysis of a discrete
  chaotic cryptosystem using external key, Physics Letters A 319~(3-4) (2003)
  334--339.

\bibitem{Li:AttackingCNSNS2008}
C.~Li, S.~Li, G.~\'{A}lvarez, G.~Chen, K.-T. Lo, Cryptanalysis of a chaotic
  block cipher with external key and its improved version, Chaos, Solitons \&
  Fractals 37~(1) (2008) 299--307.

\bibitem{Li:AttackingIVC2009}
C.~Li, S.~Li, M.~Asim, J.~Nunez, G.~Alvarez, G.~Chen, On the security defects
  of an image encryption scheme, Image and Vision Computing 27~(9) (2009)
  1371--1381.

\bibitem{Rhouma:BreakPareek:CNSNS09}
R.~Rhouma, E.~Solak, S.~Belghith, Cryptanalysis of a new
  substitution{-}diffusion based image cipher, Communications in Nonlinear
  Science and Numerical Simulation, in press, doi:10.1016/j.cnsns.2009.07.007
  (2009).

\bibitem{Pareek:modify:CNSNS2009}
V.~Patidar, N.~Pareek, K.~Sud, Modified substitution{-}diffusion image cipher
  using chaotic standard and logistic maps, Communications in Nonlinear Science
  and Numerical Simulation, in press, doi:10.1016/j.cnsns.2009.11.010 (2009).

\bibitem{Li:AttackingBitshiftXOR2007}
C.~Li, S.~Li, G.~\'{A}lvarez, G.~Chen, K.-T. Lo, Cryptanalysis of two chaotic
  encryption schemes based on circular bit shift and {XOR} operations, Physics
  Letters A 369~(1-2) (2007) 23--30.

\bibitem{Rukhin:TestPRNG:NIST}
A.~Rukhin, et~al., A statistical test suite for random and pseudorandom number
  generators for cryptographic applications, NIST Special Publication 800-22,
  available online at \url{http://csrc.nist.gov/rng/rng2.html} (2001).

\end{thebibliography}

\end{document}